\documentclass[a4paper, 11pt]{scrartcl}

\usepackage{amssymb}
\usepackage{amsthm}
\usepackage{tikz}
\usepackage{amsmath}
\usepackage{latexsym}
\usepackage{amsfonts}
\usepackage{color}
\usepackage{url}
\usepackage{algorithm}
\usepackage{algpseudocode}
\newenvironment{keywords}%
   {\begin{trivlist}\item[]{ \textit{Keywords}:}\ }% oder „Keywords:”
   {\end{trivlist}}

\DeclareSymbolFont{rsfscript}{OMS}{rsfs}{m}{n}
\DeclareSymbolFontAlphabet{\mathrsfs}{rsfscript}

\DeclareMathOperator{\Syn}{Syn}

\DeclareMathOperator{\Ker}{Ker}
\DeclareMathOperator{\M}{M}
\DeclareMathOperator{\End}{End}
\DeclareMathOperator{\Rad}{Rad}

\DeclareMathOperator{\Rnk}{Rk}
\DeclareMathOperator{\rk}{rk}
\DeclareMathOperator{\Fr}{Fr} 
\DeclareMathOperator{\Fs}{FSyn}

\DeclareMathOperator{\supp}{Supp}
\DeclareSymbolFont{rsfscript}{OMS}{rsfs}{m}{n}

\newtheorem{theorem}{Theorem}
\newtheorem{prop}{Proposition}
\newtheorem{defn}{Definition}
\newtheorem{lemma}{Lemma}
\newtheorem{prob}{Open Problem}

\newcommand{\la}{\langle}
\newcommand{\ra}{\rangle}

\title{A bound for the shortest reset words for semisimple synchronizing automata via the packing number}
\author{Emanuele Rodaro\\
Dipartimento di Matematica\\
  Politecnico di Milano\\
  Piazza Leonardo da Vinci, 32\\
  20133 Milano, Italy}

\date{}

\begin{document}
\maketitle

\begin{abstract}
We show that if a semisimple synchronizing automaton with $n$ states has a minimal reachable non-unary subset of cardinality $r\ge 2$, then there is a reset word of length at most $(n-1)D(2,r,n)$, where $D(2,r,n)$ is the $2$-packing number for families of $r$-subsets of $[1,n]$. 
\end{abstract}

%\keywords{one, two, three, four}

\section{Introduction}
An automaton is a tuple $\mathrsfs{A} = \la Q,\Sigma,\delta\ra$, where $Q$ is the set of states, $\Sigma$ is the finite alphabet acting on $Q$, and the function $\delta:Q\times\Sigma\to Q$ describes the action of $\Sigma$ on the set $Q$. More compactly we put $q\cdot a=\delta(q,a)$. This action naturally extends to $\Sigma^{*}$ and to the subsets of $Q$ in the obvious way. Automata are mostly used in theoretical computer science as languages recognizers, see for instance \cite{hop, Perrin}. However, the interested of such objects from their dynamical point of view is mostly motivated by the longstanding Cerny's conjecture regarding the class of synchronizing automata. If the automaton $\mathrsfs{A} $ has a word $u\in\Sigma^{*}$ sending all the states to a unique one, i.e., $q\cdot u=p\cdot u$ for all $q,p\in Q$, then $\mathrsfs{A}$ is called \emph{synchronizing} (or reset) and the word $u$ is called reset (or synchronizing). Cerny's conjecture states that a synchronizing automaton with $n$ states has always a reset word of length at most $(n-1)^{2}$, see \cite{Ce64}. In \cite{Ce64} it is also shown that this bound is tight by exhibiting an infinite series of synchronizing automata $\mathrsfs{C}_{n}$ having a shortest synchronizing word of length $(n-1)^{2}$. For more information
on synchronizing automata we refer the reader to Volkov's survey \cite{Vo_Survey}. The literature around Cerny's conjecture and synchronizing automata is vast and span from the algorithmic point of view to the proof of Cerny's conjecture or the existence of quadratic bounds on the smallest reset word for several classes of automata, see for instance \cite{AR16, AnVo, BeSz, BeBePe, Dubuc, GreKi, Kari, Ste11, Trah, VoChain}. The best upper bound for the shortest reset word is cubic $(n^{3}-n)/6$ obtained by Pin-Frankl \cite{Pin} and recently asymptotically improved to roughly $O(114n^{3}/685)$ by Szykula  \cite{Szy17}. 
In this paper we follow a representation theoretic approach to synchronizing automata initially pursued in \cite{AMSV09,AS09,AS07,Dubuc,Ste11}, but from a more ring theoretic point of view as followed in \cite{AR16}. We provide a new upper bound on the shortest reset word for the quite broad class of semisimple synchronizing automata. This class contains the natural class of simple synchronizing automata, i.e., automata without non-trivial congruences. 
This bound depends on the notion of former-rank of the automaton $\mathrsfs{A}$ which is the smallest reachable subset $H$ of $Q$ with $|H|>1$. The tool used is the Wedderburn-Artin theorem together with the notion of the $t$-packing number $D(t,r,n)$, i.e., the maximum size of a collection of $r$-subsets of $[1,n]$ such that no $t$-subset is covered more than once. Our main result is the following. 
\begin{theorem}
If an automaton $\mathrsfs{A}$ with $n$ states is semisimple with former-rank $r\ge 2$, then there is a reset word of length at most $(n-1)D(2,r,n)$. In particular, we have that there is a reset word of length at most
$$
\frac{n(n-1)^{2}}{r(r-1)}
$$
\end{theorem}

\section{The Wedderburn-Artin point of view}\label{sec: factoring}
In this section we fix the notation and we recall some basic facts that will be used throughout all the paper. The notation introduced and the considered approach strictly follow the one introduced in \cite{AR16}.
Henceforth, we consider a synchronizing automaton $\mathrsfs{A}=\la Q,\Sigma,\delta\ra$ with set of $n$ states $Q=\{q_1,\ldots, q_n\}$, and by $\mathcal{S}$ (sometimes also $\Syn(\mathrsfs{A})$) we denote the set of the synchronizing (or reset) words of $\mathrsfs{A}$. It is a well known fact that this set is a two-sided ideal of $\Sigma^{*}$, i.e.,  $\Sigma^{*}\mathcal{S}\Sigma^{*}\subseteq \mathcal{S}$. By $\M(\mathrsfs{A})$ we denote the transition monoid of $\mathrsfs{A}$ and by $\pi:\Sigma^*\rightarrow \M(\mathrsfs{A})$  we denote the associated natural epimorphism. Put $\mathrsfs{A}^*=\M(\mathrsfs{A})/\pi(\mathcal{S})$. There is a natural and well known action of $\M(\mathrsfs{A})$ on the set $Q$ given by $q\cdot
\pi(u)=\delta(q,u)$; we often omit the map $\pi$ and we use the simpler notation $q\cdot u$. This action extends to the subsets of $Q$ in the obvious way. By this action, $\M(\mathrsfs{A})$ embeds into the ring $\mathbb{M}_n(\mathbb{C})$ of $n\times n$ matrices with entries in $\mathbb{C}$ and with a slight abuse of notation we still denote by $\pi:\Sigma^*\rightarrow \mathbb{M}_n(\mathbb{C})$ the representation induced by this embedding. This representation determines an action of $\Sigma^*$ on
the vector space $\mathbb{C}Q$ defined by $v\cdot u= v\pi(u)$. Consider the vector $w=q_1+\cdots + q_n$ formed by summing all the elements of the canonical basis. It is a well known fact in the literature, that $\Sigma^*$ acts on the orthogonal space $w^{\perp}=\{u\in \mathbb{C}Q: \la u|w\ra=0\}$, and $u\in\mathcal{S}$ if and only if for every $v\in w^{\perp}$ we get $v\cdot u=0$ (see for instance \cite{AS07}). This induces a representation $\varphi:\Sigma^*/\mathcal{S}\rightarrow \End(w^{\perp})\simeq \mathbb{M}_{n-1}(\mathbb{C})$ with $\varphi(\Sigma^*/\mathcal{S})\simeq \mathrsfs{A}^*$. Thus, $\mathrsfs{A}^*$ may be seen as a finite multiplicative submonoid of $\mathbb{M}_{n-1}(\mathbb{C})$.
We now consider the $\mathbb{C}$-subalgebra $\mathcal{R}$ of $\mathbb{M}_{n-1}(\mathbb{C})$ generated by $\mathrsfs{A}^*$. Clearly $\mathcal{R}$ is a finitely generated
$\mathbb{C}$-algebra. Since $\mathrsfs{A}^*$ embeds into $\mathcal{R}$, with a slight abuse of notation we identify $\mathrsfs{A}^*$ with the image of this embedding $\mathrsfs{A}^*\hookrightarrow \mathcal{R}$. Therefore, the \emph{radical} $\Rad(\mathrsfs{A}^*)$ of $\mathrsfs{A}^*$ is defined by
$$
\Rad(\mathrsfs{A}^*)=\Rad(\mathcal{R})\cap
\mathrsfs{A}^*
$$
where $\Rad(\mathcal{R})$ is the radical (see \cite{Lam}) of the $\mathbb{C}$-subalgebra $\mathcal{R}$. Throughout the paper we consider the morphism 
$$
\rho:\Sigma^*\rightarrow \mathrsfs{A}^*
$$ 
that is the composition of the Rees morphism $\Sigma^*\rightarrow\Sigma^*/\mathcal{S}$ with the representation map $\varphi$. Using this map we may define the set of the radical words of the automaton $\mathrsfs{A}$ as the following ideal of $\Sigma^{*}$:
$$
\Rad(\mathrsfs{A})=\rho^{-1}(\Rad(\mathrsfs{A}^*))
$$
This is an ideal containing $\mathcal{S}$, moreover $\Rad(\mathrsfs{A})/\mathcal{S}$ is the largest nilpotent left (right) ideal of $\Sigma^*/\mathcal{S}$, see \cite{AR16}. The importance of radical words stem from the fact that if one is able to find a radical word $u$, then a synchronizing word may be obtained by considering a suitable power of $u$. Indeed, for any $u\in \Rad(\mathrsfs{A})$ it is easy to check that $u^{n-1}$ is reset. Actually, if one is able to find short reset words, then it is also possible to find a reset word of quadratic bound.
\begin{prop}\cite{AR16}
If it is true that for any strongly connected synchronizing automaton with $n$ states there is a radical word of length at most $(n-1)^{2}$, then for any strongly connected synchronizing automata with $n$ states there is a synchronizing word $u$ with $|u|\le 2(n-1)^{2}$.
\end{prop}
Finding ``short'' radical words might be an easier task than finding short reset words, thus this problem may be an intermediate step to tackle Cerny's conjecture. This statement is justified by the nice representation of the ring $\overline{\mathcal{R}}=\mathcal{R}/\Rad(\mathcal{R})$ due to the Wedderburn-Artin Theorem. Since $\overline{\mathcal{R}}$ is semisimple,  $\overline{\mathcal{R}}$ may be factorized into $k$ simple components:
\begin{equation}\label{eq: Wedderburn-Artin decomposition}
\overline{\mathcal{R}}\simeq \mathbb{M}_{n_1}(\mathbb{C}_1)\times\ldots\times \mathbb{M}_{n_k}(\mathbb{C}_k)
\end{equation}
for some (uniquely determined) positive integers $n_1,\ldots, n_k$. Let $\varphi_i:\overline{\mathcal{R}}\rightarrow \mathbb{M}_{n_i}(\mathbb{C})$, for $i\in [1,k]=\{1,\ldots, k\}$, be the projection map onto the $i$-th simple component, and let $\psi:\mathcal{R}\rightarrow \overline{\mathcal{R}}$ be the canonical epimorphism. For each $i\in[1,k]$ we defined the following frequently used morphism
$$
\theta_i= \varphi_i\circ \psi\circ \rho:\Sigma^{*}\to \mathbb{M}_{n_i}(\mathbb{C})
$$ 
Note that a radical word $u\in\Rad(\mathrsfs{A})$ may be characterized by the property $\theta_{i}(u)=0_{i}$ for each $i\in [1,k]$, where $0_{i}$ is the zero of $\mathbb{M}_{n_i}(\mathbb{C})$. We now recall a notion that plays an important role in this paper. Let $\mathcal{M}_i=\theta_i( \Sigma^*)$ be the subsemigroup of $\mathbb{M}_{n_i}(\mathbb{C})$ generated by $\Sigma^{*}$, $i\in [1,k]$. We call $\mathcal{M}_i$ the $i$-th factor monoid. The following lemma holds.
\begin{lemma}\cite{AR16}\label{lem: 0-simple semigroup}
  The $i$-th factor monoid $\mathcal{M}_i$ has a unique $0$-minimal ideal $\mathcal{I}_i$ which is a $0$-simple semigroup. Furthermore, $\mathcal{M}_i$ acts faithfully on both left and right of $\mathcal{I}_i$. 
\end{lemma}
For a word $u\in\Sigma^*$, the \emph{rank} of $u$ is $\rk(u)=|Q\cdot u|$. We recall that for any $u,v\in\Sigma^{*}$ $\rk(uv)\le \min\{\rk(u), \rk(v)\}$ holds. The notion of rank may be extended to elements in $\mathcal{I}_i\setminus\{0\}$ by defining for any $g\in \mathcal{I}_i\setminus\{0\}$ the $i$-th rank of $g$ as the following integer
$$
  \Rnk_i(g)=\min\{\rk(u):\theta_i(u)=g\}
$$
and by extension we put 
$$
  \Rnk(\mathcal{I}_i)=\min\{\Rnk_i(g):g\in\mathcal{I}_i\setminus\{0\}\}
$$
The $i$-th rank is the same for all the non-zero elements in the unique $0$-minimal ideal $\mathcal{I}_i$.
\begin{lemma}\cite{AR16}\label{lem: constant rank}
For any $g\in \mathcal{I}_i\setminus\{0\}$ we have $\Rnk_i(g)=\Rnk(\mathcal{I}_i)$.
\end{lemma}

\subsection{Simple and semisimple synchronizing automata}
An automaton-congruence (or simply a congruence) is an equivalence relation $\sigma$ on the set of states $Q$ such that $q\sigma p$ implies that $(q\cdot u) \sigma (p\cdot u)$ for all $u\in\Sigma^{*}$. The set of congruences forms a lattice with maximum the universal relation, and minimum the identity relation. In case this lattice is formed by just these two congruences the automaton is called \emph{simple}, see for instance \cite{AR16,Bab,Ry15,Thie}. For example, automata having some letters acting like a primitive group, the Cerny's automata $\mathrsfs{C}_{n}$, or some of the ``slowly synchronized'' automata $\mathrsfs{W}_{n},\mathrsfs{D}'_{n}$ considered in \cite{AVG12}, are all simple \cite{AR16}. Moreover, in a possible proof of Cerny's conjecture by induction on the number of states, simple synchronizing automata would constitutes the base case. Simple synchronizing automaton are framed nicely in the approach we are considering. We say that a synchronizing automaton $\mathrsfs{A}$ is  \emph{semisimple} whenever $\Rad(\mathrsfs{A}^*)=\{0\}$ \cite{AR16}. Equivalently, $\mathrsfs{A}$ is semisimple if and only if $\Rad(\mathrsfs{A})=\mathcal{S}$. The following result nicely frames the simple class in the class of semisimple.   
\begin{theorem}\cite{AR16}
  A synchronizing simple automaton is also semisimple.
\end{theorem}
Therefore, it seems that there is a connection between semisimplicity and the difficulty of synchronizing an automaton with a short (below the quadratic bound) reset word. On the other hand, finding reset words in the semisimple case, looks easier because of the nice structure (\ref{eq: Wedderburn-Artin decomposition}) and the fact that finding radical words is the same as finding reset words. Moreover, in case the automaton is not semisimple, there is a natural congruence that allows the construction of reset words in an ``inductive way''. The key lemma is the following. 
\begin{lemma}\cite{AR16}
  Let $\mathrsfs{A}=\la Q,\Sigma,\delta\ra$ be a synchronizing automaton which is not semisimple. Let $w\in\Rad(\mathrsfs{A})\setminus\mathcal{S}$ and let $\Ker(w)$ be the kernel of the transformation induced by the word $w$. Then there is a non-trivial congruence $\sigma$ with $\sigma\subseteq \Ker(g)$.
\end{lemma}
Indeed, if $\mathrsfs{A}$ is not semisimple and $w\in\Rad(\mathrsfs{A})$ is a radical word, then by the above lemma we may consider the quotient automaton $\mathrsfs{B}=\mathrsfs{A}/\sigma$. Consider any reset word $u\in \Syn(\mathrsfs{B})$, then since $\sigma\subseteq \Ker(w)$ we deduce that $uw\in \Syn(\mathrsfs{A})$.

\subsection{Former-rank and semisimple automata}
The \emph{former-rank} of $\mathrsfs{A} = \la Q,\Sigma,\delta\ra$ is the smallest reachable subset $H$ of $Q$ with $|H|>1$. In formulae:
$$
\Fr(\mathrsfs{A} )=\min\{|Q\cdot u|: u\in\Sigma^{*}\setminus\Syn(\mathrsfs{A}) \}
$$
The set of \emph{former-synchronizing words} is $\Fs(\mathrsfs{A} )=\{u\in\Sigma^{*}: |Q\cdot u|=\Fr(\mathrsfs{A} )\}$.
The following proposition relates the notion of former-rank with the $i$-th rank.
\begin{prop}\label{prop: former rank}
With the above notation, 
$$
\Fr(\mathrsfs{A} )\le \min\{\Rnk(\mathcal{I}_i): i\in [1,k]\}.
$$ 
Moreover, if $\Fs(\mathrsfs{A} )\setminus \Rad(\mathrsfs{A} )\neq \emptyset$ then
$$
\min\{\Rnk(\mathcal{I}_i): i\in [1,k]\}= \Fr(\mathrsfs{A} )
$$
\end{prop}
\begin{proof}
By the definition of former-rank we clearly have:
$$
\Fr(\mathrsfs{A} )\le \min\{\Rnk(\mathcal{I}_i): i\in [1,k]\}.
$$ 
Let $u\in \Fs(\mathrsfs{A} )\setminus \Rad(\mathrsfs{A} )$. Then, there is an index $i\in [1,k]$ such that $\theta_{i}(u)\neq 0_{i}$. Moreover, by Lemma \ref{lem: 0-simple semigroup}  the unit $1_i$ of the $i$-th factor monoid $\mathbb{M}_{n_i}(\mathbb{C})$ is a linear combination $\sum_j \lambda_{j}r_j=1_i$ of elements $r_{j}\in \mathcal{I}_i$ for some $\lambda_{j}\in\mathbb{C}$. Hence, 
$$
\sum_j \lambda_{j}r_j\theta_{i}(u)=\theta_{i}(u)\neq 0
$$
from which we deduce that $r_j\theta_{i}(u)\neq  0_{i}$ for some index $j$ in the summation. Hence, since $r_{j}, r_j\theta_{i}(u)\in \mathcal{I}_i$ by Lemma \ref{lem: constant rank} we have the other side of the inequality
$$
 \min\{\Rnk(\mathcal{I}_i): i\in [1,k]\}\le \Rnk(\mathcal{I}_i)=\Rnk(r_j\theta_{i}(u))\le \Rnk(\theta_{i}(u))\le \Fr(\mathrsfs{A} )
$$
\end{proof}
Note that in case the automaton is semisimple, condition $\Fs(\mathrsfs{A} )\setminus \Rad(\mathrsfs{A} )\neq \emptyset$ is always satisfied.

\section{Support and minimal sections}
Following \cite{AR16}, for a word $v\in\Sigma^{*}$ the \emph{support} $\supp(v)$ is the following subset of $[1,k]$:
$$
\supp(v)=\{i\in[1,k]: \theta_{i}(v)\neq 0_{i}\}
$$
Note that $\supp(v)=\emptyset$ if and only if $u\in\Rad(\mathrsfs{A})$. We consider the following poset
$$
\mathcal{S}(v)=\{\supp(z): z\in\Sigma^{*}v\Sigma^{*} \}
$$
ordered by inclusion. A non-empty minimal element $\supp(u)$ in $\mathcal{S}(v)$ is called a \emph{$v$-minimal section} (or just a minimal section when $v$ is clear from the context), and the element $u$ realizing such set is called \emph{$v$-minimal}. Clearly, if $\supp(u)$ is a $v$-minimal section and $\supp(z)\subsetneq \supp(u)$, for some $z\in \Sigma^{*}v\Sigma^{*}$, then $z\in \Rad(\mathrsfs{A})$. We have the following lemma.
\begin{lemma}\label{lem: remark minimal}
Let $u$ be a $v$-minimal word. Then, for any $a,b\in\Sigma^{*}$, the word $aub$ either belongs to $\Rad(\mathrsfs{A})$, or $aub$ is $v$-minimal. 
\end{lemma}
\begin{proof}
It follows from the fact that if $\theta_{j}(u)=0_{j}$, then $\theta_{j}(aub)=0_{j}$ as well. Hence, $\supp(aub)\subseteq \supp(u)$, and so since $u$ is $v$-minimal we get that either $aub\in \Rad(\mathrsfs{A})$ or $\supp(aub)=\supp(u)$.
\end{proof}
\begin{defn}
Let $u$ be a $v$-minimal word. Let $i\in\supp(u)$ and $g\in \mathcal{I}_{i}$. We say that  a word $w$ $u$-represents $g$ if $w\in\Sigma^{*}u\Sigma^{*}$, $\theta_{i}(w)=g$ and either $g=0_{i}$, or the rank $\rk(w)$ is minimum among all the words with the above properties. 
\end{defn}
We will see that the last condition is equivalent to request  that $\rk(w)=\Rnk(w)$. We have the following observation.
\begin{lemma}
If $w$ $u$-represents $g\in \mathcal{I}_{i}$, then either $g=\theta_{i}(w)=0_{i}$, or $\supp(w)=\supp(u)$.
\end{lemma}
\begin{proof}
Since $u$ is a $v$-minimal word, and $w$ contains $u$ as a factor then by Lemma \ref{lem: remark minimal} either $w\in\Rad(\mathrsfs{A})$ (corresponding to the case $g=\theta_{i}(w)=0_{i}$), or $\supp(w)=\supp(u)$.
\end{proof}
\begin{lemma}\label{lem: constant rank for representatives}
With the above notation, the following facts hold:
\begin{itemize}
\item Every element $g\in\mathcal{I}_{i}$ is $u$-representable;
\item If $w$ is a word that $u$-represents $g\in \mathcal{I}_{i}\setminus\{0_{i}\}$, then $\rk(w)=\Rnk(\mathcal{I}_i)$;
\item If $x$ is a word such that $\theta_{i}(x)=0_{i}$ for some $i\in\supp(u)$, then $\theta_{j}(x)=0_{j}$ for all $j\in \supp(u)$.
\end{itemize}
\end{lemma}
\begin{proof}
Let $u$ be a $v$-minimal word. Since $R=\mathbb{M}_{n_i}(\mathbb{C})$ is simple, then $R\theta_{i}(u)R=R$. In particular we have 
$$
\sum_{j}\lambda_{j}\theta_{i}(a_{j}ub_{j})=1_{i}
$$
for some suitable words $a_{j}, b_{j}$.
Let $z$ be a word such that $\theta_{i}(z)\in\mathcal{I}_{i}\setminus\{0_{i}\}$ and with $\rk(z)=\Rnk(\mathcal{I}_{i})$. Thus, we have:
$$
\sum_{j}\lambda_{j}\theta_{i}(a_{j}ub_{j}z)=\theta_{i}(z)\neq 0_{i}
$$
from which we deduce that there is some element $\theta_{i}(a_{s}ub_{s}z)\in\mathcal{I}_{i}\setminus\{0_{i}\}$. We clearly have 
$$
\Rnk(\mathcal{I}_{i})\le \rk(a_{s}ub_{s}z)\le \rk(z)=\Rnk(\mathcal{I}_{i})
$$
i.e., $\rk(a_{s}ub_{s}z)=\Rnk(\mathcal{I}_{i})$. If $g=0_{i}$, then $g$ is $u$-representable. Otherwise, consider a generic $g\in  \mathcal{I}_{i}\setminus\{0_{i}\}$. Since $ \mathcal{I}_{i}\setminus\{0_{i}\}$ is a $\mathcal{J}$-class, we have that there are suitable words $x,y$ such that $g=\theta_{i}(xa_{s}ub_{s}zy)$. Hence, also in this case we get
$$
\Rnk(\mathcal{I}_{i})\le\rk(xa_{s}ub_{s}zy)\le \rk(z)=\Rnk(\mathcal{I}_{i})
$$ 
and so $g$ is $u$-representable. Moreover, we have $\Rnk(\mathcal{I}_{i})\le \rk(w)\le \rk(xa_{s}ub_{s}zy)=\Rnk(\mathcal{I}_{i})$, i.e., $\rk(w)=\Rnk(\mathcal{I}_i)$.
Let us prove the last property. Take any $g\in\mathcal{I}_{j}$, for some $j\in\supp(u)$. Since $g$ is $u$-represented $g=\theta_{j}(aub)$, for some $a,b\in\Sigma^{*}$. Consider the word $aubx$, we clearly have $\theta_{i}(aubx)=0_{i}$, whence $\supp(aubx)\subsetneq \supp(aub)=\supp(u)$. Hence, since $\supp(u)$ is a minimal section we get $\supp(aubx)=\emptyset$, i.e., $g\theta_{j}(x)=\theta_{j}(aubx)=0_{j}$. Whence, $\theta_{j}(x)=0_{j}$ since $g$ is an arbitrary element in $\mathcal{I}_{j}$.
\end{proof}
We have the following lemma.
\begin{lemma}\label{lem: representative}
Let $w$ be a word that $u$-represents $g\in \mathcal{I}_{i}$, then for any $a,b\in\Sigma^{*}$ we have that the word $awb$ $u$-represents $\theta_{i}(a)g\theta_{i}(b)$. 
\end{lemma}
\begin{proof}
The following facts hold: $\theta_{i}(a)g\theta_{i}(b)\in\mathcal{I}_{i}$, $awb\in\Sigma^{*}u\Sigma^{*}$, and $\theta_{i}(awb)=\theta_{i}(a)g\theta_{i}(b)$. If $\theta_{i}(awb)=0_{i}$, then clearly $awb$ $u$-represents $\theta_{i}(a)g\theta_{i}(b)$. Otherwise, if $\theta_{i}(awb)\neq 0_{i}$, then $\theta_{i}(w)\neq 0_{i}$ and by Lemma \ref{lem: constant rank for representatives} we have:
$$
\Rnk(\mathcal{I}_i)\le \rk(awb)\le \rk(w)=\Rnk(\mathcal{I}_i).
$$
Hence, $\rk(aub)=\Rnk(\mathcal{I}_i)$ and if we would have a word $z$ containing $u$ as a factor such that $\theta_{i}(a)g\theta_{i}(b)=\theta_{i}(z)$ with $\rk(z)<\rk(awb)=\Rnk(\mathcal{I}_{i})$, then by Lemma \ref{lem: constant rank for representatives} we would have $\theta_{i}(z)=0_{i}$, a contradiction. 
\end{proof}
The following proposition shows that minimal sections are disjoint.
\begin{prop}\label{prop: empty intersection}
Let $\supp(u_{1}), \supp(u_{2})$ be two minimal sections in $\mathcal{S}(v_{1}), \mathcal{S}(v_{2})$, respectively ($v_{1}, v_{2}$ non-necessarily distinct). Then, $\supp(u_{1})\cap\supp(u_{2})=\emptyset$. 
\end{prop}
\begin{proof}
Suppose contrary to our claim that $\supp(u_{1})\cap\supp(u_{2})\neq\emptyset$. Let $i\in \supp(u_{1})\cap\supp(u_{2})$. By the same argument of Lemma \ref{lem: constant rank for representatives} since 
$$
\sum_{j}\lambda_{j}\theta_{i}(a_{i}u_{1}b_{j})=1_{i}
$$
we have have:
$$
\sum_{j}\lambda_{j}\theta_{i}(a_{i}u_{1}b_{j}u_{2})=\theta_{i}(u_{2})\neq 0_{i}
$$
Thus, we deduce that $\theta_{i}(a_{j}u_{1}b_{j}u_{2})\neq 0_{i}$ for some $j$. Since $a_{j}u_{1}b_{j}u_{2}\in \Sigma^{*}u_{1}\Sigma^{*}\cap \Sigma^{*}u_{2}\Sigma^{*}$, we get:
$$
\emptyset\neq \{s: \theta_{s}(a_{j}u_{1}b_{j}u_{2})\neq 0_{s}\}\subseteq \supp(u_{1})\cap\supp(u_{2})
$$
that contradicts the minimality of both $\supp(u_{1})$ and $\supp(u_{2})$. 
\end{proof}
We say that a subset $T\subseteq [1,k]$ is a \emph{core} whenever the condition $\theta_i(u)=0_{i}$ for all $i\in T$ implies $\theta_i(u)=0_{i}$ for all $i\in[1,k]$. Let $\mathcal{C}\subseteq [1,k]$ be a minimal core with respect to the inclusion. We have the following lemma.
\begin{lemma}\label{lem: minimal core}
Let $\mathcal{C}\subseteq [1,k]$ be a minimal core. Then, there is a family 
$$
\mathcal{F}=\{\supp(u_{1}), \ldots, \supp(u_{m})\}
$$
of minimal sections covering $\mathcal{C}$. 
\end{lemma}
\begin{proof}
Let $v_{1}$ be a word such that $\theta_{i}(v_{1})\neq 0_{i}$ for some $i\in \mathcal{C}$ and consider a $v_{1}$-minimal word $u_{1}$ such that $\supp(u_{1})$ is a minimal section. By Lemma \ref{lem: constant rank for representatives} and the definition of core we deduce $\supp(u_{1})\cap \mathcal{C}\neq \emptyset$ for if we would have $u_{1}\in\Rad(\mathrsfs{A})$ and $\supp(u_{1})=\emptyset$, a contradiction. If $\mathcal{C}\subseteq \supp(u_{1})$, then we are done. Otherwise, by the minimality of $\mathcal{C}$ we may find a word $v_{2}$ with $\theta_{i}(v_{2})=0_{i}$ for all $i\in \supp(u_{1})\cap \mathcal{C}$ such that $\theta_{j}(v_{2})\neq 0_{j}$ for some $j\in \mathcal{C}\setminus \supp(u_{1})$. Let $u_{2}$ be a $v_{2}$-minimal word such that $\supp(u_{2})$ is a minimal section. By Proposition \ref{prop: empty intersection} we have $\supp(u_{2})\cap \supp(u_{1})=\emptyset$ and by the same reason of $u_{1}$ we have that $\supp(u_{2})\cap \mathcal{C}\neq \emptyset$. If $\mathcal{C}\subseteq \left(\supp(u_{1})\cup \supp(u_{2})\right)$, then we are done. Otherwise, we may repeat (at most $|\mathcal{C}|$-times) the previous argument until we find $m$ words $u_{1}, \ldots, u_{m}$ such that $\mathcal{C}\subseteq \left(\supp(u_{1})\cup \supp(u_{2})\cup\ldots \cup \supp(u_{m})\right)$ for some minimal sections $\supp(u_{i})$, $i\in [1,m]$.
\end{proof}

\section{Main result}
In this section we prove the main result of the paper, but first we introduce an equivalence relation that is a key ingredient to prove this result. 
Let $u$ be a $v$-minimal word. Fix an index $i\in \supp(u)$. By Lemma \ref{lem: constant rank for representatives} all the elements from $\mathcal{I}_{i}$ are $u$-representable. We define a binary relation $\sigma_{i}$ on $\mathcal{I}_{i}$ in the following way. We say that $g\sigma_{i} f $ if one of the following conditions is satisfied:
\begin{itemize}
\item $g=f$;
\item there exists $w_{1},w_{2}\in\Sigma^{*}$ that $u$-represent $g$ and $f$, respectively with: 
$$|Q\cdot w_{1} \cap Q\cdot w_{2}|>1$$
\end{itemize}
We have the following proposition.
\begin{prop}\label{prop: compatibility}
The relation $\sigma_{i}$ is right-compatible with respect to the product: if $g\sigma_{i} f $, then for any $a\in \Sigma^{*}$ we have $g\theta_{i}(a)\,\sigma_{i}\,f\theta_{i}(a)$.
\end{prop}
\begin{proof}
Clearly if $f=g$ then $g\theta_{i}(a)\,\sigma_{i}\,f\theta_{i}(a)$. Otherwise, let $w_{1},w_{2}\in\Sigma^{*}$ be two words that $u$-represent $f$ and $g$, respectively, and satisfying the inequality $|Q\cdot w_{1} \cap Q\cdot w_{2}|>1$. Note that the set $H=Q\cdot w_{1} \cap Q\cdot w_{2}$ has at least two elements, and $H\cdot b\subseteq Q\cdot (aw_{1}b) \cap Q\cdot (aw_{2}b)$. We consider the following two cases.
\begin{itemize}
\item If $|H\cdot b|>1$ then $|Q\cdot (w_{1}a) \cap Q\cdot (w_{2}a)|>1$, and by Lemma \ref{lem: representative} we have that $w_{1}a$ and $w_{2}a$ $u$-represent $g\theta_{i}(a)$ and $f\theta_{i}(a)$, respectively. Hence, $g\theta_{i}(a)\,\sigma_{i}\,f\theta_{i}(a)$ holds. 
\item Otherwise from $|H\cdot b|=1$ we deduce
$$
|Q\cdot w_{1}a|< |Q\cdot w_{1}|, \; |Q\cdot w_{2}a|< |Q\cdot w_{2}|
$$
If $f,g\neq 0_{i}$, then by Lemma \ref{lem: constant rank for representatives} we necessarily have $|Q\cdot w_{1}|=|Q\cdot w_{2}|=\Rnk(\mathcal{I}_i)$, from which we deduce $g\theta_{i}(a)=f\theta_{i}(a)=0_{i}$, i.e., $g\theta_{i}(a)\,\sigma_{i}\,f\theta_{i}(a)$. If $f=g=0_{i}$, or just one among $f,g$, say $f$, is not equal to $0_{i}$, then by the same argument we deduce $g\theta_{i}(a)=f\theta_{i}(a)=0_{i}$, and so also in this case we get $g\theta_{i}(a)\,\sigma_{i}\,f\theta_{i}(a)$.
\end{itemize}
\end{proof}
Since $\sigma_{i}$ is reflexive, symmetric and right-compatible with respect to the product, we may consider the transitive closure $\sim_{i}$ of $\sigma_{i}$ that is clearly a right-congruence on $\mathcal{I}_{i}$.
\\
To state the next result we need to recall some basic fact on the packing problem \cite{HandComb}. Let $X=[1,n]$ be a finite set of $n$ elements, and let $t,r$ be two integers in $[1,n]$. The $t$-packing problem is the problem of determining the maximum size $D(t,r,n)$ of a collection of $r$-subsets of $X$ such that no $t$-subset is covered more than once. With a double counting argument one can easily show that the following upper bound holds:
\begin{equation}\label{eq: packing bound}
D(t,r,n)\le \frac{{n \choose t}}{{r \choose t}}
\end{equation}
with equality if and only if a Steiner system $S(t,r,n)$ exists. Note that if $r_{1}\le r_{2}$, then $D(t,r_{1},n)\ge D(t, r_{2},n)$. 
\\
Henceforth we put $r_{i}=\Rnk(\mathcal{I}_i)$, $i\in [1,k]$, and $n=|Q|$. We have the following proposition. 
\begin{prop}\label{prop: packing number}
With the above notation we have that $|\mathcal{I}_i/\sim_{i}|\le D(2,r_{i},n)+1$. In particular, for any $g\in \mathcal{I}_{i}$ there is a word $z$ with $|z|\le D(2,r_{i},n)$ such that 
$$
g\theta_{i}(z)\sim_{i} 0_{i}
$$
\end{prop}
\begin{proof}
Let $\mathcal{I}_i/\sim_{i}=\{[0_{i}]_{\sim_{i}}, [g_{1}]_{\sim_{i}},\ldots, [g_{\ell}]_{\sim_{i}}\}$. Let $z_{1}, \ldots, z_{\ell}$ be words that $u$-represent $g_{j}$ for $j\in [1,\ell]$. Put $F_{j}=Q\cdot z_{j}$ for $j\in [1,\ell]$. Then, by Lemma \ref{lem: constant rank for representatives} $\mathcal{F}=\{F_{1}, \ldots, F_{\ell}\}$ is a family of $r_{i}$-sets satisfying the property that 
$$
|F_{i}\cap F_{j}|\le 1
$$
for all $i,j\in [1,\ell]$ with $i\neq j$. Therefore, each pair is covered at most once, and so $\ell\le D(2,r_{i},n)$. Since $\sim_{i}$ is a right-congruence, $\mathcal{M}_{i}$ acts on the right of $\mathcal{I}_i/\sim_{i}$. Thus, for any $g\in \mathcal{I}_{i}$ there is a word $z$ such that $[g]_{\sim_{i}}\theta_{i}(z)=[0_{i}]_{\sim_{i}}$, and so since $\ell\le D(2,r_{i},n)$ we may find such $z$ with $|z|\le \ell$.
\end{proof}
We have the following proposition. 
\begin{prop}\label{prop: single component}
For any $v$-minimal word $u$ and $i\in\supp(u)$, there is a word $w_{i}\in\Sigma^{*}$ with $|w_{i}|\le n_{i}D(2,r_{i},n)$ such that
$$
\theta_{i}(w_{i})=\sum_{j=1}^{m}h_{j}t_{j}
$$
for some elements $h_{j}\in\mathbb{M}_{n_i}(\mathbb{C})$ and $t_{j}\in \mathcal{I}_{i}$ with $t_{j}\sim_{i} 0_{i}$ for all $j\in[1,m]$.
\end{prop}
\begin{proof}
Let $R=\mathbb{M}_{n_i}(\mathbb{C})$. Since $R$ is simple there are elements $g_1,\ldots, g_m\in \mathcal{I}_i$ such that
$$
Rg_{1}+\cdots+ Rg_{m}=R
$$
Moreover, since $R$ is the direct sum of $n_i$ left ideals, by the Jordan-H\"{o}lder Theorem we may assume $m\le n_j$. In particular we have 
$$
\sum_{j=1}^{m} h_{j}g_j=1_i
$$
for some suitable elements $h_{j}\in R$. Consider $g_{1}\in  \mathcal{I}_{i}$, by Proposition \ref{prop: packing number} there is a word $z_{1}$ such that $|z_{1}|\le  D(2,r_{i},n)$ and $g_{1}\theta_{i}(z_{1})\sim_{i} 0_{i}$. Consider $g_{2}\theta(v_{1})\in \mathcal{I}_{i}$ and apply again Proposition \ref{prop: packing number} to find a word $z_{2}$ such that $|z_{2}|\le  D(2,r_{i},n)$ and $g_{2}\theta(z_{1}z_{2})\sim_{i} 0_{i}$. Since $\sim_{i}$ is a right-congruence, $g_{1}\theta_{i}(z_{1}z_{2})\sim_{i} 0_{i}$ holds as well. Continuing in this way we may find a sequence $z_{1}, \ldots, z_{s}$ of words such that each $|z_{j}|\le  D(2,r_{i},n)$ and 
$$
g_{j}\theta_{i}(z_{1}z_{2}\ldots z_{s})\sim_{i} 0_{i}
$$
holds for any $s\le m$. In particular taking $w_{i}=z_{1}\ldots z_{m}$ we have that $|w_{i}|\le  mD(2,r_{i},n)\le  n_{i}D(2,r_{i},n)$ and $g_{j}\theta_{i}(z_{1}z_{2}\ldots z_{m})\sim_{i} 0_{i}$ for all $j\in[1,m]$. Therefore, by putting $t_{j}=g_{j}\theta_{i}(w_{i})$, $j\in [1,m]$, we deduce that 
$$
\theta_{i}(w_{i})=\sum_{j=1}^{m}  h_{j}t_{j}
$$
for some suitable $t_{j}\in \mathcal{I}_i$ with $t_{j}\sim_{i} 0_{i}$.
\end{proof}
In case the automaton is semisimple we obtain the following result.
\begin{theorem}\label{theo: zero component}
If the automaton $\mathrsfs{A}$ is semisimple, then for any $v$-minimal word $u$ and $i\in \supp(u)$, there is a word $w$ with  
$$
|w|\le \min_{i\in\supp(u)}\{n_{i}D(2,r_{i},n)\}
$$
such that $\theta_{i}(w)=0_{i}$ for all $i\in \supp(u)$.
\end{theorem}
\begin{proof}
By Proposition \ref{prop: single component} there is a word $w_{i}$ with $|w_{i}|\le n_{i}D(2,r_{i},n)$ such that $\theta_{i}(w_{i})=\sum_{j=1}^{m}h_{j}t_{j}$ for some elements $h_{j}\in\mathbb{M}_{n_i}(\mathbb{C})$ and $t_{j}\in \mathcal{I}_{i}$ with $t_{j}\sim_{i} 0_{i}$ for all $j\in[1,m]$. Thus, it is enough to prove that $t_{j}=0_{i}$ for each $j\in[1,m]$. Consider a generic $t_{j}$ for some $j\in [1,m]$. We may assume $t_{j}\neq 0_{i}$. Since $\sim_{i}$ is the transitive closure of $\sigma_{i}$ we have that there are $\ell> 1$ elements $f_{1}, \ldots f_{\ell}\in\mathcal{I}_{i}$ with $f_{1}=t_{j}$, $f_{\ell}=0_{i}$ such that $f_{s}\sigma_{i} f_{s+1}$ for all $s\in [1, \ell-1]$. Choosing the minimal $\ell$ we may assume that $f_{1}, \ldots f_{\ell}$ are distinct. By definition of $\sigma_{i}$ we have words $z_{1}, \ldots z_{\ell}$ that $u$-represent $f_{1}, \ldots, f_{\ell}$, respectively, and such that 
\begin{equation}\label{eq: packing condition}
|Q\cdot z_{s}\cap Q\cdot z_{s+1}|>1
\end{equation}
for all $s\in [1,\ell-1]$. 
\\
We claim that $z_{\ell}$ is reset. Indeed, since $z_{\ell}$ $u$-represents $0_{i}$ we have that $u$ is a factor of $z_{\ell}$ that is also a $v$-minimal word. Hence, $\supp(z_{\ell})\subseteq \supp(u)$ and since $\theta_{i}(z_{\ell})=0_{i}$ we have $\supp(z_{\ell})\subsetneq \supp(u)$ which by the minimality condition on $u$ implies $z_{m}\in \Rad(\mathrsfs{A})=\Syn(\mathrsfs{A})$ since $\mathrsfs{A}$ is semisimple. Suppose $\ell>1$. Therefore, by (\ref{eq: packing condition}) and and the fact that $z_{m}\in\Syn(\mathrsfs{A})$ we get:
$$
1=|Q\cdot z_{m-1}\cap Q\cdot z_{m}|>1
$$
a contradiction. Thus, $\ell=1$ and $t_{j}=f_{1}=0_{i}$. Hence, we get that $\theta_{i}(w_{i})=0_{i}$. The statement now follows by taking the index $i\in \supp(u)$ such that $n_{i}D(2,r_{i},n)$ is minimum and by the last property of Lemma \ref{lem: constant rank for representatives} 
\end{proof}
The following lemma is similar to \cite[Lemma 16]{AR16} and we state here with proof for the sake of completeness. 
\begin{lemma}\label{lem: null ideal}
Consider an ideal $I$ of $\overline{\mathcal{R}}$ of the form
$$
I=\mathbb{M}_{n_{i_1}}(\mathbb{C})\times\cdots\times \mathbb{M}_{n_{i_m}}(\mathbb{C})
$$
for some subset $T=\{i_1,\ldots, i_m\}$ of $[1,k]$. Let $J=\psi^{-1}(I)$. There is a sequence $i_{j}\in\{i_1,\ldots, i_m\}$ of integers for $j\in[1,\ell]$ such that for any words $z_{j}$, $j\in [1,\ell]$ with $\theta_{i_{j}}(z_{j})=0_{i_{j}}$, the word 
$$
u=\prod_{j=\ell}^{1}z_{j}
$$
such that $\rho(u) J=0$. Moreover, $\sum_{j=1}^{\ell}n_{i_{j}}\le n-1$. 
\end{lemma}
\begin{proof}
Renumbering the indexes in the decomposition of the ideal, we may suppose without loss of generality that $I=\mathbb{M}_{n_{1}}(\mathbb{C})\times\cdots\times \mathbb{M}_{n_{m}}(\mathbb{C})$ for some $m\le k$, and so we may consider $T=[1,m]$. Since $\mathcal{R}$ is a subalgebra of $\mathbb{M}_{n-1}(\mathbb{C})$, the vector space $V=\mathbb{C}^{n-1}$ is a $J$-module. By Proposition~4.8 of~\cite{Lam} $J$ and $J/\Rad(J)$ have the same simple left modules. By Exercise~4.7 of~\cite{Lam} we have $\Rad(J)=J\cap \Rad(\mathcal{R})$, hence $J/\Rad(J)=I$. Let
$$
V=V_0\supset V_1\supset\ldots\supset V_{j}\supset\ldots\supset V_\ell=0
$$
be a Jordan-H\"older series. Each module $V_{j-1}/V_j$ for $j\in [1,\ell]$ is a simple $J$-module and so, by the above argument, also an $I$-module. In particular, we have $u v=\psi(u) v$ for all $u\in J$, $v\in V_{j-1}/V_j$.
We claim that either $mv=0$ for all $m\in J$, $v\in V_{j-1}/V_j$ or there is a $v\in V_{j-1}/V_j$ such that $V_{j-1}/V_j=\mathbb{M}_{n_{i}}(\mathbb{C})v$, for some $i\in T$ and $n_{i}=\dim_{\mathbb{C}}(V_{j-1}/V_j)$, where $\dim_{\mathbb{C}}(V_{j-1}/V_j)$ is the dimension of the
$\mathbb{C}$-vector space $V_{j-1}/V_j$. Indeed, the first condition occurs only if for every $v\in V_{j-1}/V_j$, $mv=0$ for all $m\in\mathbb{M}_{n_{i}}(\mathbb{C})$ and for all $i\in T$. Otherwise, we may assume that $mv\neq 0$, for some $v\in V_{j-1}/V_j$ and $m\in \mathbb{M}_{n_{i}}(\mathbb{C})$ for some $i\in T$. Thus, $\mathbb{M}_{n_{i}}(\mathbb{C}) v$ is a
left $I$-submodule of $V_{j-1}/V_j$ which is non-trivial, thus
$V_{j-1}/V_j=\mathbb{M}_{n_{i}}(\mathbb{C})v$. Therefore, $V_{j-1}/V_j$ is a simple $\mathbb{M}_{n_{i}}(\mathbb{C})$-module and by Theorem~3.3 of~\cite{Lam} $n_{i}=\dim_{\mathbb{C}}(V_{j-1}/V_j)$. Hence, putting $i_{j}=i$, if $z_{j}$ is any word such that $\theta_{i}(z_{j})=0_{i}$ we deduce
$$
  \rho(z_j)V_{j-1}/V_j=\psi(\rho(z_j))\mathbb{M}_{n_{i}}(\mathbb{C})v=\theta_{i}(z_{j})\mathbb{M}_{n_{i_j}}(\mathbb{C})v=0
$$
In case $mv=0$ for all $m\in J$, $v\in V_{j-1}/V_j$ we clearly have that also $\rho(z_j)V_{j-1}/V_j=0$ holds. Therefore, the following word
$$
u=z_{\ell}\ldots z_{1}=\prod_{j=\ell}^{1}z_{j}
$$
satisfies $\rho(u)JV=0$ since $\rho(z_j)V_{j-1}/V_j=0$ for all $j\in[1,\ell]$. Thus, $\rho(u)J=0$. Moreover, we have 
$$
\sum_{j=1}^{\ell}n_{i_{j}}=\sum_{j=1}^\ell \dim_{\mathbb{C}}(V_{j-1}/V_j)\le \dim_{\mathbb{C}}(V)=n-1
$$
\end{proof} 
\begin{theorem}\label{theo: bound for semisimple}
With the above notation. If the automaton $\mathrsfs{A}$ is semisimple, and $\mathcal{C}$ is a minimal core of $\mathrsfs{A}$, then there is a reset word $w$ with  
$$
|w|\le (n-1)\max_{i\in\mathcal{C}}\{D(2,r_{i},n)\}
$$
\end{theorem}
\begin{proof}
By a suitable permutation we may assume $\mathcal{C}=[1,m]$ for some $m\le k$. By Lemma \ref{lem: minimal core} there is a family 
$$
\mathcal{F}=\{\supp(u_{1}), \ldots, \supp(u_{t})\}
$$
of minimal sections covering $\mathcal{C}$. Consider the ideal pinpointed by $\mathcal{C}$
$$
I=\mathbb{M}_{n_{1}}(\mathbb{C})\times\cdots\times \mathbb{M}_{n_{m}}(\mathbb{C})
$$
By Theorem \ref{theo: zero component} for each $j\in[1,t]$ there is a word $w_{j}$ with  
$$
|w_{j}|\le \min_{i\in\supp(u_{j})}\{n_{i}D(2,r_{i},n)\}
$$
such that $\theta_{i}(w_{j})=0_{i}$ for all $i\in \supp(u_{j})$. Thus, by Lemma \ref{lem: null ideal} we may find a sequence of integers $i_{j}\in [1, t]$, $j\in[1,\ell]$ such that the following word
$$
u=\prod_{j=\ell}^{1} w_{i_{j}}
$$
satisfies the property $\rho(u)J=0$ for $J=\psi^{-1}(I)$. Hence, $\theta_{i}(u)=0_{i}$ for all $i\in [1,m]$, and since $\mathcal{C}$ is a core and $\mathrsfs{A}$ is semisimple we conclude that $u\in\Rad(\mathrsfs{A} )=\Syn(\mathrsfs{A})$. Moreover, we get
$$
|u|\le \max_{i\in\mathcal{C}}\{D(2,r_{i},n)\}\sum_{j=1}^{\ell}n_{i_{j}}\le (n-1)\max_{i\in\mathcal{C}}\{D(2,r_{i},n)\}
$$
\end{proof}
From the previous theorem we immediately obtain our main result. 
\begin{theorem}\label{theo: main theo}
If an automaton $\mathrsfs{A}$ with $n$ states is semisimple with former-rank $r=\Fr(\mathrsfs{A})$, then there is a reset word of length at most $(n-1)D(2,r,n)$. In particular, we have that there is a reset word of length at most
$$
\frac{n(n-1)^{2}}{r(r-1)}
$$
\end{theorem}
\begin{proof}
From Theorem \ref{theo: bound for semisimple} we deduce that there is a reset word of length at most
$$
(n-1)\max_{i\in\mathcal{C}}\{D(2,r_{i},n)\}
$$
Let $r=\min\{r_{i}, i\in [1,k]\}$. By Proposition \ref{prop: former rank} we have $r=\Fr(\mathrsfs{A})$. Moreover, since $D(2,r_{i},n)\le D(2,r,n)$ for all $i\in [1,k]$ we may conclude that there is a reset word of length at most $(n-1)D(2,r,n)$. By the upper bound (\ref{eq: packing bound}) we immediately get the bound in the statement. 
\end{proof}

\section{Conclusion and open problems}
The bound $n(n-1)^{2}/r(r-1)$ of Theorem \ref{theo: main theo} is already better than the Pin-Frankl's bound for $r\ge 3$, but not asymptotically better than Szykula's $O(114n^{3}/685)$. However, it starts to be better already for $r\ge 4$. Moreover, $n(n-1)^{2}/r(r-1)$ is a straightforward upper bound for $ D(2,r,n)$ although they are asymptotic \cite{HandComb}. There is a slightly more precise bound for $ D(2,r,n)$ and many others for specific choices of the parameters, we remind the reader to \cite[Section 14]{HandComb} for further details. For instance, for ``small'' $n$ 
$$
D(2,r,n)\le \frac{(r-1)n}{r^{2}-n}
$$
holds provided that the denominator is positive. Thus, if $r^{2}-n\ge 0$ we may conclude that there is a reset word of length $(n-1)n(r-1)/(r^{2}-n)$. 
\\
Semisimplicity is just used in Theorem \ref{theo: zero component}. Hopefully, using similar techniques, it may possible to extend the main result to a general synchronizing automaton:
\begin{prob}
For a general synchronizing automaton $\mathrsfs{A}$ with former-rank $r$, prove that there is a radical word of length at most $(n-1)D(2,r,n)$.
\end{prob}
Even though the previous open problem would be solved, the crucial case remains that of former-rank two. Indeed, in this case $D(2,2,n)=n(n-1)/2$ and this gives rise to a non interesting upper bound for the shortest reset word. However, this case suggests that the automata that are more difficult to synchronize, are the ones having former-rank two. Therefore, the following direction of research seems important in understanding how to crack Cerny's conjecture. 
\begin{prob}
What is the structure of $i$-th factors $\mathcal{M}_{i}$ and their unique $0$-minimal ideals $\mathcal{I}_{i}$ in case of the former-rank two? 
\end{prob}
In Proposition \ref{prop: packing number} it is used the fact that each pair is covered at most once by  the family $\mathcal{F}=\{F_{1}, \ldots, F_{\ell}\}$ of $r$-sets, from which we deduced $\ell\le D(2,r,n)$.  
However, there is an action of $\Sigma^{*}$ on the family $\mathcal{F}$ therefore we may state the following ``dynamical packing problem'' in the hope to have a better upper bound. 
\begin{prob}
Find the maximum size $Dd(t,r,n,k)$ of a collection $\mathcal{F}=\{F_{1}, \ldots, F_{m}\}$ of $r$-subsets of $[1,n]$ such that no $t$-subset is covered more than once and with the property that the alphabet $[1,k]$ acts partially on $\mathcal{F}$, and this action is transitive.
\end{prob}
For instance using the Cerny's series $\mathrsfs{C}_{n}$, it is not difficult to check that in case $k=r=t=2$ we have $Dd(2,2,n,2)=D(2,2,n)={n \choose t}/{r \choose t}$.

\end{document}